\documentclass[11pt]{article}

\usepackage[letterpaper, margin=1.1in] {geometry}
\usepackage{indentfirst}
\usepackage{amsfonts}
\usepackage{amsmath}
\usepackage{amsthm}
\usepackage{color}
\usepackage{eufrak}
\usepackage{verbatim}
\usepackage[makeroom]{cancel}

\usepackage{amssymb}
\usepackage{amsthm}
\usepackage{mathrsfs}
\usepackage{epsfig}
\usepackage{makeidx}
\usepackage{graphicx}
\usepackage{indentfirst}
\usepackage[numbers]{natbib}
\usepackage{color}
\usepackage[backref=page]{hyperref}
\usepackage{graphicx}
\usepackage{hyperref}
\usepackage{graphicx}
\usepackage{amsmath}
\usepackage{pdfpages}

\definecolor{corlinks}{RGB}{0,0,150}
\definecolor{cormenu}{RGB}{0,0,150}
\definecolor{corurl}{RGB}{0,0,150}

\hypersetup{
colorlinks=true,
urlcolor=corlinks,
linkcolor=corlinks,
menucolor=cormenu,
citecolor=corlinks,
pdfborder= 0 0 0
}

\newtheorem{theorem}{Theorem}
\newtheorem{lemma}{Lemma}
\newtheorem{corollary}{Corollary}

\newtheorem{remark}{Remark}

\newtheorem{claim}{Claim}

\newcommand{\eqdef}{\stackrel{\rm def}{=}}

\def\colorful{0}

\ifnum\colorful=1

\newcommand{\blue}[1]{{\color{blue} {#1}}}

\fi
\ifnum\colorful=0

\newcommand{\blue}[1]{{{#1}}}

\fi

\begin{document}

\title{On monotone circuits with local oracles and clique lower bounds\\ \vspace{0.65cm}}

\author{Jan Kraj\'{\i}\v{c}ek   \and  Igor C. Oliveira}

\date{{\small Faculty of Mathematics and Physics}\\{\small
Charles University in Prague}\\~\\ \normalsize \vspace{-0.6cm} ~\\~\\ \today}

\maketitle

\begin{abstract}
We investigate monotone circuits with local oracles [K., 2016], i.e., circuits containing additional inputs $y_i = y_i(\vec{x})$ that can perform unstructured computations on the input string $\vec{x}$. Let $\mu \in [0,1]$ be the locality of the circuit, a parameter that bounds the combined strength of the oracle functions $y_i(\vec{x})$, and $U_{n,k}, V_{n,k} \subseteq \{0,1\}^m$ be the set of \blue{$k$-cliques and the set of complete $(k-1)$-partite graphs, respectively (similarly to [Razborov, 1985]).} Our results can be informally stated as follows. 
\begin{itemize}
\item[(\emph{i})] For an appropriate extension of depth-$2$ monotone circuits with local oracles, we show that the size of the smallest circuits separating $U_{n,3}$ (triangles) and $V_{n,3}$ (complete bipartite graphs) undergoes two phase transitions according to $\mu$. 
\item[(\emph{ii})] For $5 \leq k(n) \leq n^{1/4}$, arbitrary depth, and $\mu \leq 1/50$, we prove that the monotone circuit size complexity of separating the sets $U_{n,k}$ and $V_{n,k}$ is $n^{\Theta(\sqrt{k})}$, under a certain restrictive assumption on the local oracle gates. 
\end{itemize}

The second result, which concerns  monotone circuits with restricted  oracles, extends and provides a matching upper bound  for the  exponential lower bounds on the monotone circuit size complexity of $k$-clique obtained by Alon and Boppana (1987). 
\end{abstract}

\section{Introduction and motivation}

We establish initial lower bounds on the power of monotone circuits with local oracles (monotone CLOs), an extension of monotone circuits introduced in \citep{1611.08680}  motivated by  problems in proof complexity. Interestingly, while the model has been conceived as part of an approach to establish new length-of-proofs lower bounds,  our results indicate that investigating such circuits can benefit our understanding of classical results obtained in the usual setting of monotone circuit complexity, where no oracle gates are present \blue{(see the discussion on the Alon-Boppana exponential lower bounds for $k$-clique \citep{DBLP:journals/combinatorica/AlonB87} presented later in this section)}. 

Before describing the circuit model and our contributions in more detail, which require no background in proof complexity, we explain the main motivation that triggered our investigations.\\

\noindent \textbf{Relation to proof complexity.} A major open problem in proof complexity is to obtain lower bounds on proof length in $F_d[\oplus]$, depth-$d$ Frege systems extended with parity connectives (cf.~\citep{krajicek}). It is known that strong enough lower bounds for $F_{3}[\oplus]$, the depth-$3$ version of this system, imply related lower bounds for each system $F_d[\oplus]$, where $d \in \mathbb{N}$ is arbitrary \citep{buss2015collapsing}. A natural restriction of $F_{3}[\oplus]$ for which proving general lower bounds is still open is the proof system R$(\mathsf{Lin}/\mathbb{F}_2)$ (cf.~\citep{ItsykonSokolov}, \citep{1611.08680}). It corresponds to an extension of Resolution where clauses involve linear functions over $\mathbb{F}_2$.\footnote{Lower bonds for tree-like R$(\mathsf{Lin}/\mathbb{F}_2)$-proofs were  established in \citep{ItsykonSokolov}.} 

In order to attack this and other related problems, \citep{1611.08680} proposed a generalization of the feasible interpolation method to randomized feasible interpolation. Among other results, \citep{1611.08680} established that lower bounds on the size of monotone circuits with local oracles separating the sets $U_{n,k}$ and $V_{n,k}$ (defined below) imply lower bounds on the size of general (dag-like) R$(\mathsf{Lin}/\mathbb{F}_2)$ proofs. In addition, it was shown that strong lower bounds in the new circuit model would provide a unifying approach to important length-of-proofs lower bounds established via feasible interpolation (cf.~\citep[Section 6]{1611.08680}, \citep{pudlak1997lower}).\\

Motivated by these connections and by the important role of feasible interpolation in proof complexity, we start in this work a more in-depth investigation of the power and limitations of monotone circuits with local oracles. We focus on the complexity of the $k$-clique problem over the classical sets of negative and positive instances considered in monotone circuit complexity \citep{razborov1985lower, DBLP:journals/combinatorica/AlonB87}. While the monotone complexity of $k$-clique has been investigated over other input distributions of interest (cf.~\citep{DBLP:journals/siamcomp/Rossman14}), we remark that the structure of these instances is particularly useful in proof complexity (cf.~\citep{Krajicek97,  pudlak1997lower, DBLP:journals/jsyml/BonetPR97}). The corresponding tautologies have appeared in several other works.

We provide next a brief introduction to the circuit model and to the set of instances of $k$-clique that are relevant to our results.\\

\noindent \textbf{An extension of monotone circuits.} A monotone circuit with local oracles $C(\vec{x}, \vec{y})$ is a monotone boolean circuit containing extra inputs $y_j$ (local oracles) that compute an arbitrary \emph{monotone} function of $\vec{x}$. In order to limit the power of these oracles, there is a locality parameter $\mu \in [0,1]$ that controls the sets of positive and negative inputs on which the inputs $y_i$ can be helpful. In more detail, we consider circuits computing a monotone function $f \colon \{0,1\}^m \to \{0,1\}$, and associate to each input $y_i$ a rectangle $U_i \times V_i$, with $U_i \subseteq f^{-1}(1)$ and $V_i \subseteq f^{-1}(0)$. We restrict attention to sets of rectangles whose union have measure at most $\mu$ according to an appropriate distribution $\mathcal{D}$ that depends on $f$. We are guaranteed that $y_i(U_i) = 1$ and $y_i(V_i) = 0$ but, crucially, the computation of $C(\vec{x}, \vec{y})$ must be correct no matter the interpretation of each $y_i$ outside its designated sets $U_i$ and $V_i$.\\

\noindent \textbf{The $k$-clique function and the sets $U_{n,k}$ and $V_{n,k}$.} We focus on the monotone boolean function $f \colon \{0,1\}^m \to \{0,1\}$ that outputs $1$ on an $n$-vertex graph $G \in \{0,1\}^m$ if and only if it contains a clique of size $k$, where $m = \binom{n}{2}$. More specifically, we investigate its complexity as a partial boolean function over $U_{n,k} \cup V_{n,k}$, where $U_{n,k}$ is the set of inputs corresponding to $k$-cliques over the set $[n]$ of vertices, and $V_{n,k}$ is \blue{the set of complete $\zeta$-partite graphs over $[n]$, where $\zeta = k-1$.} Roughly speaking, for this choice of $f$, we measure the size of a subset $\mathcal{B}  \subseteq U_{n,k} \times V_{n,k}$ using the product distribution obtained from the uniform distribution over the $k$-cliques in $U_{n,k}$, and the distribution supported over $V_{n,k}$ obtained by sampling a random coloring $\chi \colon [n] \to [k-1]$ of $[n]$ using \blue{exactly $\zeta = k-1$ colors, and considering the associated complete $\zeta$-partite graph $G(\chi)$.}\footnote{\blue{Some authors consider as negative instances the larger set of complete $\zeta$-partite graphs where $\zeta$ ranges from $1$ to $k-1$. For technical reasons, we work with exactly $(k-1)$-partite graphs (cf.~Claim \ref{c:Fmonotone}). In most lower bound contexts this is inessential, as a random coloring $\chi \colon [n] \to [k-1]$ under a bounded $k(n)$ contains non-empty color classes except with an exponentially small probability.}}\\

A more rigorous treatment of the circuit model and of the problem investigated in our work appears in Section \ref{s:notation}.

\subsection{Our Results}

We observe a phase transition for an extension of depth-2 monotone circuits with local oracles that separate triangles from complete bipartite graphs. 

\begin{theorem}[Phase transitions in depth-2]\label{t:thm1}
Let $s = s(n,\mu)$ be the minimum size of a depth-$2$ monotone circuit \emph{(DNF)} on inputs $\vec{x}$, $y_i(\vec{x})$, and $g_j(\vec{y})$ that separates $U_{n,3}$ and $V_{n,3}$, where the $y$-inputs have locality $\leq \mu$, and each $g_j$ is an arbitrary monotone function on $\vec{y}$. Then, for every $\varepsilon > 0$,
 \begin{equation}\nonumber
    s \;=\;
    \begin{cases}
      ~1 & \text{if}~~\mu = 1, \\
      ~\Theta_\varepsilon(n^2) & \text{if}~~1/2 + \varepsilon \;\leq\; \mu \;\leq\; 1 - \varepsilon, \\
            ~\Theta_\varepsilon(n^3) & \text{if}~~0 \;\leq\; \mu \;\leq\; 1/2 - \varepsilon.
    \end{cases}
  \end{equation}
Furthermore, the upper bounds on $s(n,\mu)$ do not require the extra inputs $g_j(\vec{y})$.
\end{theorem}

\blue{Observe that the lower bounds remain valid in the presence of the functions $g_j(\vec{y})$. In other words, in the restricted setting of depth-2 circuits, a small locality parameter does not help, even if arbitrary monotone computations that depend on the output of the local oracle gates are allowed in the circuit. (As explained in Section \ref{s:thm1}, the monotone functions $g_j(\vec{y})$ can be handled in a generic way, and add no power to the model.)} 
 
The proof of Theorem \ref{t:thm1} is presented in Section \ref{s:thm1}.  The argument considers different bottlenecks in the computation based on the value of $\mu$. In our opinion, the main conceptual message of Theorem \ref{t:thm1} is that an interesting complexity-theoretic behavior appears already at depth two. Indeed, the oracle gates can interact with the standard input variables in unexpected ways, and the main difficulty when analyzing general monotone CLOs is the arbitrary nature of these gates, which are limited only by the locality parameter.\footnote{It is plausible that the analysis behind the proof of Theorem \ref{t:thm1} extends to larger $k$, but we have not pursued this direction in the context of depth-2 circuits. See also the related discussion on Section \ref{s:concluding}.}

We obtain stronger results for larger $k = k(n)$ and with respect to unrestricted monotone circuits (i.e., arbitrary depth), but our approach requires an extra condition on the set of rectangles that appear in the definition of the oracle gates. \blue{Our assumption, denoted by $\mathcal{A}_d$, says that} if each oracle variable $y_i$ is associated to the rectangle $U_i \times V_i$, then the intersection of every collection of $d+1$ sets $U_i$ is empty.

\begin{theorem}[Upper and lower bounds for monotone circuits with restricted oracles]\label{t:thm2}~\\ 
For every $k = k(n)$ satisfying $5 \leq k \leq n^{1/4}$, the following holds.
\begin{itemize}
\item[\emph{1.}] If $D(\vec{x}, \vec{y})$ is a monotone circuit with local oracles that separates $U_{n,k}$ and $V_{n,k}$ and its $y$-variables have locality $\mu \leq 1/16$ and satisfy condition $\mathcal{A}_d$, then $\mathsf{size}(D) = n^{\Omega(\sqrt{k}/d)}$.
\item[\emph{2.}] For every $\varepsilon > 0$, there exists a monotone circuit with local oracles $C(\vec{x}, \vec{y})$ of size $n^{O_\varepsilon(\sqrt{k})}$ separating $U_{n,k}$ and $V_{n,k}$ whose $y$-variables have locality $\mu \leq \varepsilon$ and satisfy condition $\mathcal{A}_1$.
\end{itemize}
\end{theorem}

\vspace{0.2cm}

The proof of Theorem \ref{t:thm2} appears  in Section \ref{s:thm2}. The lower bound extends results on the monotone circuit size complexity of $k$-clique for large $k = k(n)$ obtained in \citep{DBLP:journals/combinatorica/AlonB87}.\footnote{For $k \leq \log n$, near-optimal results were proved in \citep{razborov1985lower}.} Indeed, our argument relies on their analysis of Razborov's approximation method \citep{razborov1985lower}, with extra work required to handle the oracle gates.  The upper bound is achieved by an explicit description of a monotone CLO generalizing the construction from Theorem \ref{t:thm1}. \blue{The following corollary, stated for reference, is immediate from Theorem \ref{t:thm2}.}

\begin{corollary}\label{c:tight}
Let $5 \leq k(n) \leq n^{1/4}$, $\mu = 1/50$, and assume rectangles are mapped to local oracle gates in a way that no $k$-clique is associated to more than a constant number of rectangles. Then  the monotone circuit size complexity of separating the sets $U_{n,k}$ and $V_{n,k}$ is $n^{\Theta(\sqrt{k})}$.
\end{corollary}

\blue{(We note that the constant $1/50$ appearing in this statement is not particularly important, and that any small enough constant locality parameter $\mu$ suffices.) To our knowledge, Corollary \ref{c:tight} provides the first explanation for the  tightness of the Alon-Boppana \citep{DBLP:journals/combinatorica/AlonB87} exponential lower bounds for $k$-clique. In particular, in order to prove monotone circuit lower bounds for this problem stronger than $n^{\sqrt{k}}$ in the regime where $k(n) \gg \mathsf{poly}(\log n)$, one has to consider either a different set of instances, or employ a technique that does not apply to circuits with local oracles of constant locality.}\footnote{\blue{We remark that much tighter monotone lower bounds of the form $n^k/\mathsf{poly}(\log n)$ are known in the regime where $k$ is constant or slightly super-constant \citep{razborov1985lower, DBLP:journals/combinatorica/AlonB87}. Interestingly, these results do not generalize to circuits with local oracles due to the different choice of parameters employed in the corresponding legitimate lattices.}}

We discuss some directions for future investigations in Section \ref{s:concluding}, where we also say a few more words on the connection to proof complexity.\footnote{We have made no attempt to optimize the constants and the asymptotic notation appearing in Theorems \ref{t:thm1} and \ref{t:thm2}.}

\section{Notation and basic facts}\label{s:notation}

Let $[e]$ denote the set $\{1, 2, \ldots, e\}$, $e \in \mathbb{N}$. For a set $B$, we use $\binom{B}{\ell}$ to denote the family of subsets of $B$ of size exactly $\ell$. The function $\log (\cdot)$ refers to logarithm in base $2$. For a set $V$, we use $v \sim V$ to denote a uniformly distributed element from $V$.  We are interested in the computation of partial boolean functions over $\{0,1\}^m$. For $A \subseteq \{0,1\}^m$, a function $f \colon A \to \{0,1\}$ is monotone if $x, y \in A$ and $x \preceq y$ (i.e, $x_i \leq y_i$ for all $i \in [m]$) imply $f(x) \leq f(y)$.\\

\noindent \textbf{Monotone CLOs.}  A monotone boolean circuit $C(x_1, \ldots, x_n, y_1, \ldots, y_e)$ on $n$ variables and $e$ local oracles (monotone CLO for short) is a (non-empty) directed acyclic graph containing $\leq n + e + 2$ sources and one sink (the output node). The non-source nodes have in-degree $2$. Source nodes are labeled by elements in $\{x_1, \ldots, x_n\} \cup \{y_1, \ldots, y_e\} \cup \{0,1\}$, and each non-source node is labeled by a gate symbol in $\{\wedge, \vee\}$. We say that $C$ has size $s$ if the total number of nodes in the underlying graph is $s$, including source nodes. The computation of $C$ on an input string $(a,b) \in \{0,1\}^n \times \{0,1\}^e$ is defined in the natural way.

The formulation above is consistent with the statement of Theorem \ref{t:thm2}. In Theorem \ref{t:thm1}, which concerns bounded-depth circuits, we allow the internal $\{\wedge, \vee\}$-nodes to have unbounded fan-in.

We consider the computation of $C(\vec{x}, \vec{y})$ on input pairs where each bit in the second input $\vec{y}$ is a function of $\vec{x}$. Furthermore, we will restrict our analysis to monotone computations over a set $A \subseteq \{0,1\}^n$ of interest. For this reason, to specify the computation of $C$ on a string $x \in A$, we will associate to each local oracle variable $y_i$ a corresponding monotone function $f_i \colon A \to \{0,1\}$. 

In order to obtain a non-trivial notion of circuit complexity in this model, we use a real-valued parameter $\mu \in [0,1]$ to control the family of admissible functions $f_i$. Each function $f_i$ separates a particular pair of sets $U_i \subseteq f^{-1}(1) \subseteq A$ and $V_i \subseteq f^{-1}(0) \subseteq A$, but $C$ must be correct no matter the choice of the functions $f_i$ separating these sets.  The parameter $\mu$ captures the measure of $\bigcup_i U_i \times V_i$. This is formalized by the definitions introduced next.\\

\noindent \textbf{Correctness and locality.} 
Let $f \colon \{0,1\}^n \to \{0,1\}$, $U \subseteq f^{-1}(1)$, $V \subseteq f^{-1}(0)$, $W = (U,V)$, and $A = U \cup V$. Moreover, let $U_1, \ldots, U_e \subseteq U$ and $V_1, \ldots, V_e \subseteq V$ be sets of inputs, and for convenience, let $\mathcal{W} = (W_i)_{i \in [e]}$ denote the sequence of pairs $W_i = (U_i, V_i)$. Finally, let $\mathcal{D}$ be a probability distribution supported over $U \times V$. We say that $\mathcal{W}$ has locality $\mu$ with respect to $\mathcal{D}$ if, for $\mathcal{B} = \bigcup_{i \in [e]} U_i \times V_i$,
$$
\Pr_{(u,v) \sim \mathcal{D}}\big [(u,v) \in \mathcal{B} \big ] \;\leq\; \mu.
$$

We say that a pair $W' = (U', V')$ is included in the pair $W = (U,V)$ if $U' \subseteq U$ and $V' \subseteq V$, and that a sequence $\mathcal{W}  = (U_i, V_i)_{i \in [e]}$ of pairs is included in $W$ if each member $W_i = (U_i,V_i)$ of $\mathcal{W}$ is included in $W$. Let $g \colon A \to \{0,1\}$ be an arbitrary monotone boolean function over $A = U \cup V$. We say that $g$ separates a pair $(U',V')$ if $g(U') = 1$ and $g(V') = 0$. Let $\mathcal{F} = (f_1, \ldots, f_e)$ be a sequence of functions, where each $f_i \in A \to \{0,1\}$ is monotone. We say that $\mathcal{F}$ separates $\mathcal{W}$ if each $f_i$ separates $(U_i, V_i)$. For convenience, we also say in this case that $\mathcal{F}$ is a $\mathcal{W}$-separating sequence of functions.

Given a monotone CLO pair $(C, \mathcal{W})$ as above, and a  $\mathcal{W}$-separating sequence $\mathcal{F}$ of monotone functions, let $$C(\vec{x}, \mathcal{F}) \eqdef C(x_1, \ldots, x_n, f_1(\vec{x}), \ldots, f_e(\vec{x}))$$ denote the function in $A \to \{0,1\}$ that agrees with the output of $C$ when each oracle input $y_i$ is set to $f_i(x)$. Observe that $C(x, \mathcal{F})$ is a \emph{monotone} function over $A = U \cup V$, since $C$ is a monotone circuit and each $f_i$ is a monotone function over $A$. We will sometimes abuse notation and view $C(x, \mathcal{F})$ as a circuit. We say that the pair $(C, \mathcal{W})$ computes the function $f \colon A \to \{0,1\}$ if for every $\mathcal{W}$-separating sequence $\mathcal{F}$ of monotone functions, we have $C(a,  \mathcal{F}) = f(a)$ for all $a \in A$. \blue{(We stress that the monotone CLO pair must be correct on every input string, and on every $\mathcal{W}$-separating sequence.)}

Finally, let $f \in \{0,1\}^n \to \{0,1\}$ be a monotone function, $A = U \cup V$ for sets $U \subseteq f^{-1}(1)$ and $V \subseteq f^{-1}(0)$, and $W = (U,V)$. We say that $f$ can be computed over $A \subseteq \{0,1\}^n$ by a monotone circuit with local oracles of size $s$ and locality $\mu$ (with respect to a distribution $\mathcal{D}$) if there exists a monotone circuit $C(\vec{x}, \vec{y})$ of size $\leq s$ and a sequence $\mathcal{W} = (U_i,V_i)_{i \in [e]}$ of length $e \leq s$ that is included in $W$ and has locality $\leq \mu$ such that the monotone CLO pair $(C,\mathcal{W})$ computes $f$ over $A$.

For convenience of notation, we will sometimes write $y_i = y[U_i,V_i]$ to indicate a local oracle over the pair $W = (U_i, V_i)$.\\

\noindent \textbf{Defining $U_{n,k}$, $V_{n,k}$, and $\mathcal{D}_{n,k}$.} Let $m = \binom{n}{2}$, where $n \geq 4$, and let $k \in \mathbb{N}$ be an integer satisfying $3 \leq k < n$. We view $[n]$ as a set of vertices, and $[m]$ as its associated set of (undirected) edges. For $B \subseteq [n]$, we use $K_B \in \{0,1\}^m$ to denote the graph (also viewed as a string) corresponding to a clique over $B$. Let
\begin{eqnarray} 
U_{n,k} & \eqdef & \Big \{K_B \in \{0,1\}^m \mid B \in  \binom{[n]}{k} \Big \},~\text{and}\nonumber \\ 
V_{n,k} & \eqdef & \{H \in \{0,1\}^m \mid H~\text{is a non-trivial complete}~\zeta\text{-partite graph},~\blue{\text{where}~\zeta = k-1}\}, \nonumber \\
A_{n,k} & \eqdef & U_{n,k} \cup V_{n,k}. \nonumber
\end{eqnarray}
Clearly, $U_{n,k} \cap V_{n,k} = \emptyset$. It is convenient to associate to each coloring $\chi \colon [n] \to [k-1]$ a corresponding graph $G(\chi)$, where $e = \{v_1,v_2\} \in E(G(\chi))$ if and only if $\chi(v_1) \neq \chi(v_2)$. Let
$$
V_{n,k}^\chi \eqdef \{\chi \mid \chi \colon [n] \to [k-1]\}
$$
be the family of all possible colorings of $[n]$ using at most $k-1$ colors. \blue{Under our definitions, for a given coloring $\chi \in V^{\chi}_{n,k}$ we have $G(\chi) \in V_{n,k}$ if and only if $|\chi([n])| = k - 1$.} We measure the locality of monotone CLO pairs $(C, \mathcal{W})$ separating $U_{n,k}$ and $V_{n,k}$ with respect to a product distribution $\mathcal{D}_{n,k} \eqdef \mathcal{D}_{n,k}^U \times \mathcal{D}_{n,k}^V$, whose components are defined as follows. $\mathcal{D}_{n,k}^U$ is simply the uniform distribution over the $k$-cliques in $U_{n,k}$, while $\mathcal{D}_{n,k}^V$ assigns to each fixed graph $H \in V_{n,k}$ probability mass
$
\mathcal{D}_{n,k}^V(H) \;\eqdef\; \Pr_{\chi \sim V_{n,k}^\chi}[G(\chi) = H \mid G(\chi) \in V_{n,k}]
$.\footnote{\blue{Note that the probability that a random coloring $\chi \colon [n] \to [k-1]$ contains less than $k-1$ non-trivial color classes is exponentially small in $n$ for the values of $k(n)$ investigated in Theorems \ref{t:thm1} and \ref{t:thm2}.}} \blue{(This is simply the uniform distribution over $V_{n,k}$, but this is not the most convenient point of view in some estimates.)}\\

\noindent \textbf{The sequence $\mathcal{F}^\star$.} The definition introduced above agrees with the formulation of monotone circuits with oracles from \citep{1611.08680}. We stress that a source of difficulty when computing a function $f \colon A \to \{0,1\}$ using a monotone circuit $C(\vec{x}, \vec{y})$ and a sequence $\mathcal{W} = (W_i)$ of pairs included in $W = (f^{-1}(0), f^{-1}(1))$ is that $C(x, \mathcal{F})$ must be correct for \emph{every} $\mathcal{W}$-separating sequence $\mathcal{F} = (f_i)$ of monotone functions. In order to prove lower bounds against a monotone CLO pair $(C,\mathcal{W})$, we will consider a particular instantiation of the monotone functions $f_i \colon A \to \{0,1\}$, discussed next.

Let $y_i = y[U_i,V_i]$ be a local oracle variable associated with the pair $W_i = (U_i,V_i)$. We define the function $f^{\star}_{W_i} \colon A \to \{0,1\}$ as follows:
$$
  f^{\star}_{W_i}(x) =
  \begin{cases}
    1 & \text{if $x \in U_i \cup (V \,\backslash\, V_i)$,}\\
    0 & \text{otherwise.}
  \end{cases}
$$
Observe that $f^\star_{W_i}(U_i) = 1$ and $f^\star_{W_i}(V_i) = 0$. In particular, $f^\star_i \eqdef f^\star_{W_i}$ separates the pair $W_i$. We use $\mathcal{F}^\star \eqdef (f^\star_i)$ to denote the corresponding sequence of functions for a given choice of $\mathcal{W} = (W_i)$.

For an arbitrary monotone function $f \colon A \to \{0,1\}$, $U_i \subseteq U \subseteq f^{-1}(1)$, and $V_i \subseteq V \subseteq f^{-1}(0)$, $f^\star_i$ is not necessarily monotone. However, for the problem investigated in our work $f^\star_i$ is always monotone, as stated next.

\begin{claim}\label{c:Fmonotone}
Let $3 \leq k < n$. For every pair $W_i = (U_i, V_i)$ with $U_i \subseteq U_{n,k}$ and $V_i \subseteq V_{n,k}$, the function $f^\star_i \colon A_{n,k} \to \{0,1\}$ is monotone.
\end{claim}

\begin{proof}
It is enough to observe that, under these assumptions, there are no distinct strings $a_1, a_2 \in A_{n,k}$ satisfying $a_1 \preceq a_2$. \blue{Here we crucially used that the $(k-1)$-partite graphs in $V_{n,k}$ have exactly $k-1$ non-empty parts.}
\end{proof}

The use of $\mathcal{F}^{\star}$ to prove lower bounds against monotone CLO pairs $(C, \mathcal{W})$ computing a monotone function $f \colon A \to \{0,1\}$ is justified by the following observation, which describes an extremal property of $\mathcal{F}^\star$.

\begin{claim}\label{cl:hardF}
Let $\mathcal{F} = (f_i)$ be an arbitrary $\mathcal{W}$-separating sequence of monotone functions $f_i \colon A \to \{0,1\}$.
If $C(x,\mathcal{F})$ is incorrect on an input $a \in A$, then $C(x,\mathcal{F}^\star)$ is also incorrect on $a$.
\end{claim}

\begin{proof}
Assume that $a \in U$. Consequently, $f(a) = 1$, and the assumption that $C(x,\mathcal{F})$ is incorrect means that $C(x,\mathcal{F}) = 0$. Using that each $f_i$ separates $W_i = (U_i,V_i)$ and the definition of $f^\star_i$, we get $f^\star_i(a) \leq f_i(a)$. By the monotonicity of the circuit $C$, it follows that $C(a, \mathcal{F}^\star) \leq C(a, \mathcal{F})$. Thus $C(a, \mathcal{F}^\star)$ is incorrect on input $a$ as well. The case where $a \in V$ is analogous.
\end{proof}

\blue{Therefore, $\mathcal{F}^\star$ is the
hardest separating-sequence, meaning that any circuit that computes $f$ under
$\mathcal{F}^\star$ computes $f$ under any separating-sequence.}

\begin{remark}[Simulating negated inputs]
It is possible to simulate negated input variables in $C$ using oracles gates. For instance, if $x_{\{1,2\}}$ corresponds to the input edge $\{1,2\}$, we define an oracle gate $y[U',V']$ with $U' = \{K_B \in U_{n,k} \mid \neg x_{\{1,2\}}(K_B) = 1\}$ and  $V' = \{H \in V_{n,k} \mid \neg x_{\{1,2\}}(H) = 0\}$. It is well-known that $U_{n,k}$ and $V_{n,k}$ can be separated by counting input edges and using a single negation gate. However, it is easy to see that, by combining the latter construction with the trick above, we get monotone circuits with oracles of huge locality.
\end{remark}

Indeed, for the problem investigated here, monotone circuits with local oracles can be seen as an intermediary model between monotone and non-monotone circuits, where the locality parameter $\mu$ restricts the computation of the extra input variables $y_i$.

In order to be precise, we rephrase the hypothesis $\mathcal{A}_d$ employed in Theorem \ref{t:thm2} using the notation introduced in this section.\\

\noindent \textbf{The assumption $\mathcal{A}_d$.} Let $d \in \mathbb{N}$, and $(C,\mathcal{W})$ be a monotone CLO pair with $\mathcal{W} = (W_i)_{i \in I}$, $W_i = (U_i, V_i)$, $U_i \subseteq U$ and $V_i \subseteq V$. We say that $(C,\mathcal{W})$ satisfies $\mathcal{A}_d$ if there exists no $u \in U$ and $I' \subseteq I$, $|I'| > d$ such that $u \in \bigcap_{i' \in I'} U_{i'}$.

\section{Phase transitions in depth-2: Proof of Theorem \ref{t:thm1}} \label{s:thm1}

Our argument relies on Claims \ref{c:Fmonotone} and \ref{cl:hardF} described in Section \ref{s:notation}. We start with a straightforward adaptation of a lemma from \citep{1611.08680}.

\begin{lemma}\label{l:lemmaK} Let $C(\vec{x},\vec{y})$ be a monotone circuit, $A = U \,\cup\, V$ be a disjoint union, $W = (U,V)$, and $\mathcal{W} = (W_i)_{i \in [e]}$ be a sequence of pairs included in $W$, where each $W_i = (U_i,V_i)$. Then,
\begin{itemize}
\item[\emph{1.}] Over inputs $a \in A$, for every $i,j \in [e]$, the following holds\emph{:}
\begin{eqnarray}
f^{\star}_{(U_i, V_i)} \vee f^{\star}_{(U_j, V_j)} & = & f^{\star}_{(U_i \cup U_j, V_i \cap V_j)}. \nonumber \\
f^{\star}_{(U_i, V_i)} \wedge f^{\star}_{(U_j, V_j)} & = & f^{\star}_{(U_i \cap U_j, V_i \cup V_j)}. \nonumber
\end{eqnarray}
\item[\emph{2.}] Let $\mathcal{B} \eqdef \bigcup_{i \in [e]} U_i \times V_i \subseteq U \times V$, and $i,j \in [e]$. Then $(U_i \cap U_j) \times (V_i \cup V_j) \subseteq \mathcal{B}$ and $(U_i \cup U_j) \times (V_i \cap V_j) \subseteq \mathcal{B}$.
\end{itemize}
\end{lemma}

\begin{proof}
Immediate from the definitions.
\end{proof}

First, we prove a weaker version of Theorem \ref{t:thm1} that forbids the extra inputs $g_j(\vec{y})$. Then we use Lemma \ref{l:lemmaK} to observe that our argument extends to the more general class of circuits.

Let $\varepsilon > 0$ be a fixed constant, and $n$ be sufficiently large.\\

\noindent \textbf{Case 1: $\mu = 1$.} Obviously, there is a trivial monotone CLO pair $(C,\mathcal{W})$ with locality $\mu = 1$ that separates $U_{n,3}$ and $V_{n,3}$: $C$ contains a single node $y_1$, and $W_1 = (U_{n,3}, V_{n,3})$.\\

\noindent \textbf{Case 2: $1/2 + \varepsilon \leq \mu \leq  1 - \varepsilon$.} We start with the upper bound. In other words, we construct a monotone CLO of size $O(n^2)$ and locality $\leq 1/2 + o(1)$.\footnote{This construction is inspired by discussions in \citep{RRobere}.} Let $x_{\{i,j\}}$ for $i \neq j \in [n]$ denote the input variable corresponding to edge $\{i,j\} \in \binom{[n]}{2}$. Consider the following monotone circuit:

$$
C(\vec{x}, \vec{y}) \;\eqdef\; \bigvee_{i < j}(x_{\{i,j\}} \wedge y_{\{i,j\}}).
$$
We associate to each $y_{\{i,j\}} = y[U_{\{i,j\}}, V_{\{i,j\}}]$ the sets\\

$U_{\{i,j\}} \eqdef \{ K_B \in U_{n,3} \mid \{i,j\} \subseteq B~\text{and these are the smallest elements in}~B\},~\text{and}~$\\

$V_{\{i,j\}} \eqdef \{ H \in V_{n,3} \mid \text{vertices}~i~\text{and}~j~\text{are in different parts of}~H\}.$\\

\noindent Observe that $C$ has size $O(n^2)$.

First, we argue that this monotone CLO is correct. If the input graph is a triangle $K_B \in \{0,1\}^m$ with $B = \{i,j,k\}$, where $i < j < k$, then $x_{\{i,j\}}(K_B) = 1$. Moreover, for any monotone function $f_{\{i,j\}}$ that separates $(U_{\{i,j\}}, V_{\{i,j\}})$, we must have $f_{\{i,j\}}(K_B) = 1$, since $K_B \in U_{\{i,j\}}$ by construction. Thus $C(K_B,\mathcal{F})$ must accept $K_B$ for all separating sequences $\mathcal{F} = (f_{\{i,j\}})$.  Now let $H \in V_{n,3}$ be a complete bipartite graph over $[n]$ with non-empty parts $V^H_1$ and $V^H_2$ partitioning $[n]$. \blue{We show that for  $i <j$ it holds that $x_{\{i,j\}}(H) \wedge y_{\{i,j\}}(H) = 0$.} If for some $x_{\{i,j\}}$ we have $x_{\{i,j\}}(H) = 1$, then $i,j$ are in different parts of $H$. By construction, any $f_{\{i,j\}}$ separating the pair $(U_{\{i,j\}}, V_{\{i,j\}})$ must output $0$ on $H$. Consequently, the output of the circuit on $H$ is $0$, under any sequence $\mathcal{F}$ of separating functions.

Next, we upper bound the locality of the $y$-variables. Let $\mathcal{B} = \bigcup_{i < j} U_{\{i,j\}} \times V_{\{i,j\}} \subseteq U_{n,3} \times V_{n,3}$. Let $(K_B, H)$ be a fixed input pair in $U_{n,3} \times V_{n,3}$. Observe that this pair is in $\mathcal{B}$ if and only if there exist $i, j \in [n]$ with $i < j$ such that:

\vspace{0.2cm}

(1) \,$\{i,j\} \in B$, 

(2) \,these are the smallest elements in $B$, and 

(3) \,the vertices $i$ and $j$ belong to different parts of $H$.

\vspace{0.2cm} 

\noindent Therefore, the locality $\mu$ of the monotone CLO defined above is upper bounded by
$$
\Pr_{(K_B, H) \sim \mathcal{D}_{n,3}}[\exists\,i<j ~\text{satisfying}~(1), (2), (3)]  \quad \leq \quad  
 \sum_{i < j} \Pr[(i,j)~\text{satisfies}~(1), (2), (3)] \nonumber
$$
\vspace{-0.37cm}
\begin{eqnarray}
(\text{using independence}) & = &  \sum_{i < j} \Pr_{H \sim \mathcal{D}^V_{n,3}}[(i,j)~\text{satisfies}~(3)]\cdot \Pr_{K_B \sim \mathcal{D}^U_{n,3}}[(i,j)~\text{satisfies}~(1),(2)] \nonumber \\
& = & \Pr_{\chi \sim V^\chi_{n,3}}[\chi(1) \neq \chi(2) \mid \chi([n])  = \{1,2\}] \cdot \sum_{i < j} \frac{n-j}{\binom{n}{3}} \nonumber\\
& = & (1/2 + o(1)) \cdot 1 \;\leq\; 1/2 + \varepsilon.\nonumber
\end{eqnarray}

\vspace{0.1cm}

We argue next the lower bound on circuit size for this range of $\mu$. In other words, we prove that if $\mu \leq 1 - \varepsilon$ then the circuit size is $\Omega_\varepsilon(n^2)$. Let $(C, \mathcal{W})$ be a monotone CLO pair, where $C(\vec{x}, \vec{y})$ is a monotone DNF with $t \leq s$ terms, $\mathcal{W} = (W_i)_{i \in [e]}$, $e \leq s$, $W_i = (U_i, V_i)$, and each $W_i$ is included in the pair $(U_{n,3}, V_{n,3})$. Further, let $\mathcal{B} = \bigcup_i U_i \times V_i$. Assume the pair $(C,\mathcal{W})$ computes $3$-clique over $A_{n,3}$. In order to establish a lower bound, we consider the sequence $\mathcal{F}^\star$, as defined in Section \ref{s:notation}. Then, using Lemma \ref{l:lemmaK}, we can write this circuit in an equivalent way as follows:
\begin{equation}\label{eq:circuitform}
C(\vec{x}, \mathcal{F}^\star) \;=\;\bigvee_{j \in [t]} \left ( \bigwedge_{e \in S_j} x_e \wedge f^\star_{(U'_j, V'_j)}(\vec{x})  \right )\,,
\end{equation}
where $S_j \subseteq \binom{[n]}{2}$ and $U'_j \times V'_j \subseteq \mathcal{B}$, for each $j \in [t]$. This is without loss of generality, since terms that did not originally include a $y$-variable can be represented using $f^\star_{(U_{n,3}, \emptyset)}$, which is equivalent to the constant $1$ function over inputs in  $A_{n,3}$.

Next, observe that if $|S_j| > 3$ for some $j \in [t]$ then the corresponding term cannot accept an input from $U_{n,3}$. Thus we can assume without loss of generality that $0 \leq |S_j| \leq 3$. Partition the terms of $C(\vec{x}, \mathcal{F}^\star)$ into sets $T_\ell$, $0 \leq \ell \leq 3$, with $T_\ell$ containing all terms for which $|S_j| = \ell$.

Every triangle $K_B$ accepted by a term from $T_0$ forces a measure $\geq 1/\binom{n}{3}$ in $\mathcal{B}$, since the corresponding functions $f^\star_{(U'_j,V'_j)}$ must satisfy $V'_j = V_{n,3}$ in order for the term not to accept a complete bipartite graph $H \in V_{n,3}$. Consequently, using that $\mu \leq 1 - \varepsilon$, a total number of at most $r = (1- \varepsilon) \binom{n}{3}$ triangles can be accepted by terms in $T_0$.

Now each term in $T_2$ or in $T_3$ accepts at most one triangle, and each term in $T_1$ accepts at most $n$ triangles. Therefore, using the preceding paragraph, in order for the circuit to accept all $\binom{n}{3}$ triangles in $U_{n,3}$, we must have:

$$|T_1| \cdot n + |T_2| + |T_3| \geq \binom{n}{3} - r = \Omega(n^3).$$
This implies that at least one of $|T_1|$, $|T_2|$, and $|T_3|$ must be $\Omega(n^2)$. In particular, the original circuit must have size at least $\Omega(n^2)$.\\

\noindent \textbf{Case 3: $0 \leq \mu \leq  1/2 - \varepsilon$.} The $O(n^3)$ size upper bound at $\mu = 0$ is achieved by the trivial monotone circuit for $3$-clique. For the lower bound, we adapt the argument presented above. Using the same notation, we assume there is a correct circuit as described in (\ref{eq:circuitform}). By the same reasoning, $|S_j| \leq 3$ for each $j \in [t]$. Furthermore, we can assume that the edges corresponding to each $S_j$ are contained in some triangle from $U_{n,3}$. 

Rewrite $C(\vec{x}, \mathcal{F}^\star)$ as an equivalent circuit $C'$:
\begin{equation}\label{eq:circuit2}
C'(\vec{x}, \mathcal{F}^\star) \;\eqdef\;\bigvee_{\ell \in I_{\leq 2}}  \left  ( \bigwedge_{e \in S_\ell} x_e \wedge f^\star_{(U_\ell,V_\ell)}(\vec{x}) \right ) \vee \bigvee_{i \in I_{3}}  \left  ( \bigwedge_{e \in S_i} x_e \wedge f^\star_{(U_i,V_i)}(\vec{x}) \right ),
\end{equation}
where $I_{\leq 2}$ contains the indexes of the original sets $S_j$ such that the edges obtained from $S_j$ touch at most $2$ vertices, and $I_{3}$ contains the indexes corresponding to sets $S_j$ whose edges span exactly $3$ vertices.

First, suppose there exists $\ell \in I_{\leq 2}$ such that $\mathcal{D}_{n,3}^V(V_\ell) \leq 1/2 - \varepsilon/4$. This implies that $f^\star_{\ell}$ rejects a subset of $V_{n,3}$ of measure at most $1/2 - \varepsilon/4$. Moreover, using that $\ell \in I_{\leq 2}$, $\bigwedge_{e \in S_\ell} x_e$ rejects a subset of $V_{n,3}$ of measure at most $1/2 + \varepsilon/8$. Consequently, the $\ell$-th term of the original circuit $C(\vec{x}, \mathcal{F}^\star)$ must accept some negative input from $V_{n,3}$. This violates the assumption that the initial monotone CLO pair computes $3$-clique over $A_{n,3}$.

We get from the previous argument that for every $\ell \in I_{\leq 2}$, $\mathcal{D}^V_{n,3}(V_\ell) \geq 1/2 - \varepsilon/4$. Consider now the quantity $\eta = |\bigcup_{\ell \in I_{\leq 2}} U_{\ell}|/|U_{n,3}|$, and observe that $\mu \geq \eta \cdot (1/2 - \varepsilon/4)$ by the previous density lower bound. Since we are in the case where $\mu \leq 1/2 - \varepsilon$, we obtain  $\eta \leq 1 - \Omega_{\varepsilon}(1)$. 

In turn, using the definition of $\eta$ and of $\mathcal{F}^\star$, it follows that the left-hand side of $C'(\vec{x}, \mathcal{F}^\star)$ in (\ref{eq:circuit2}) accepts at most a $\eta$-fraction of $U_{n,3}$. By the correctness of $C(x,\mathcal{F}^\star)$, the right-hand side of the equivalent circuit $C'(\vec{x},\mathcal{F}^\star)$ must accept at least a $\Omega_\varepsilon(1)$-fraction of the triangles in $U_{n,3}$. Now observe that for each $i \in I_3$, the corresponding term $\bigwedge_{e \in S_i} x_e$ accepts exactly one triangle. Therefore, we must have $|I_3| \geq \Omega_\varepsilon(\binom{n}{3})$. This completes the proof  that $t = \Omega(n^3)$.\\

In order to prove lower bounds in the presence of $g_j(\vec{y})$ input variables, observe that the following holds. First, all lower bounds were obtained using $\mathcal{F}^\star$. Due to Lemma \ref{l:lemmaK}, each $g_j(\vec{y})$ is equivalent over $A_{n,3}$ to $f^\star_{(U'_j,V'_j)}$, for an appropriate pair $(U'_j, V'_j)$ satisfying $U'_j \times V'_j \subseteq \mathcal{B}$. Finally, in addition to the locality bound, the inclusion in $\mathcal{B}$ is the only information about the $y$-variables that was employed in the proofs. In other words, each $g_j(\vec{y})$ can be treated as a new $y$-variable in the arguments above, without affecting the locality bounds. 

This extends the lower bound to the desired class of circuits, and completes the proof of Theorem \ref{t:thm1}.

\section{Circuits with restricted oracles: Proof of Theorem \ref{t:thm2}}\label{s:thm2}

We start with the upper bound.

\begin{lemma}\label{l:upperb}
Let \blue{$3 \leq k \leq n^{1/4}$} and $2 \leq \ell < k$. There exists a monotone circuit with local oracles $E(\vec{x}, \vec{y})$ of size $O(\binom{n}{\ell} \cdot \binom{\ell}{2})$ and locality $\mu \leq \exp(- \Omega(\ell^2/k))$ that computes $k$-clique over $A_{n,k}$. Furthermore, the local oracles associated to $E$ satisfy condition $\mathcal{A}_1$.
\end{lemma}

\begin{proof}
We generalize a construction in the proof of Theorem \ref{t:thm1}. For every set $B \in \binom{[n]}{k}$, let $F(B) \in \binom{B}{\ell}$ be the lexicographic first $\ell$-sized subset of $B$. Consider the following monotone circuit with local oracles: 
$$
E(\vec{x}, \vec{y}) \;\eqdef\; \bigvee_{D \in \binom{[n]}{\ell}} \Big ( \bigwedge_{e \in \binom{D}{2}} x_{e}  \wedge y_{D}\; \Big )\;,
$$
where to each $y_{D}$ we associate a pair $(U_D, V_D)$ with $U_{D} \times V_{D} \subseteq U_{n,k} \times V_{n,k}$, defined as follows:
\begin{center}
$U_{D} \eqdef \{ K_B \in U_{n,k} \mid F(B) = D\}$ ~ and ~ $V_{D} \eqdef \{ H \in V_{n,k} \mid K_D \subseteq H\}$.\end{center}

By construction, $U_D \cap U_{D'} = \emptyset$ for distinct $D, D' \in \binom{[n]}{\ell}$. In other words, assumption $\mathcal{A}_1$ is satisfied. Further, the size of $E$ is $O(\binom{n}{\ell} \cdot \binom{\ell}{2})$. The correctness of this monotone CLO can be established by a straightforward generalization of the argument from Section \ref{s:thm1}. It remains to estimate its locality parameter $\mu$.

Fix a set $D \in \binom{[n]}{\ell}$, and let $\gamma_D \eqdef \mathcal{D}^{V}_{n,k}(V_D)$. By symmetry, $\gamma_D = \gamma_{D'}$ for every $D' \in \binom{[n]}{\ell}$. Since distinct sets $U_{D}$ are pairwise disjoint and locality is measured with respect to the product distribution $\mathcal{D}_{n,k} = \mathcal{D}_{n,k}^U \times \mathcal{D}_{n,k}^V$, the locality of the oracle rectangles associated with $E$ is at most $\gamma_D$.  This value can be upper bounded as follows:
\begin{eqnarray}
\gamma_D = \Pr_{H \sim \mathcal{D}^V_{n,k}}[K_D \subseteq H] & = & \Pr_{\chi \sim V^\chi_{n,k}}[K_D \subseteq G(\chi) \mid G(\chi) \in V_{n,k}] \nonumber \\
& = & \frac{\Pr_\chi[ K_D \subseteq G(\chi) \wedge G(\chi) \in V_{n,k} ]}{\Pr_\chi[G(\chi) \in V_{n,k}]} \nonumber \\
& \leq & \frac{\Pr_\chi[K_D \subseteq G(\chi)]}{\Pr_\chi[\,\blue{|\chi([n])| = k-1}\,]}
\nonumber \\
(\text{using~}\blue{3 \leq k \leq n^{1/4}}\text{~and~}n \to \infty) & \leq & (1 + o(1)) \cdot \frac{(k-1)(k-2)\ldots(k-\ell)}{(k-1)^\ell} \nonumber \\
 & \leq & (1 + o(1)) \cdot \frac{(k - \lfloor \ell/2 \rfloor)^{\ell/2}}{(k-1)^{\ell/2}} \nonumber \\
& = & (1 + o(1)) \cdot \left ( 1 - \frac{\lfloor \ell/2 \rfloor - 1}{k - 1} \right )^{\ell/2} \nonumber \\ 
(\text{using~}(1 - x) \leq e^{-x}~\text{and}~0 \leq x \leq  1)& \leq & \exp(-\Omega(\ell^2/k)).\nonumber
\end{eqnarray}
This completes the proof of Lemma \ref{l:upperb}.
\end{proof}

The upper bound in Theorem \ref{t:thm2} follows immediately from Lemma \ref{l:upperb}, by taking a large enough $\ell = O(\sqrt{k})$. Observe that, more generally, one can get a trade-off between circuit size and locality.\\

We move on now to the lower bound part, which relies on a sequence of lemmas. For a set $X \subseteq [n]$, we let $\lceil X \rceil \eqdef \bigwedge_{\{i,j\} \in \binom{X}{2}} x_{\{i,j\}}$ be the corresponding clique indicator circuit. For convenience, we define $\lceil X \rceil \eqdef 1$ if $X$ is a singleton or the empty set. Also, note that $\lceil X \rceil = \lceil X \rceil \wedge f^\star_{U_{n,k}, \emptyset}$ over $A_{n,k}$. Under this notation, we don't need to consider standalone terms in the lemma below, which adapts to our setting a result from \citep{1611.08680}. 

\begin{lemma}\label{l:lb-depth2}
Let $\mathcal{W} = (W_i)$ with $W_i = (U_i,V_i)$ be a sequence of pairs included in $(U_{n,k},V_{n,k})$. Let $C(\vec{x}, \vec{y})$ be a monotone circuit with local oracles of the form
$$
C(x, y) \;=\;\bigvee_{i \in [t]} \Big ( \lceil X_i \rceil \wedge y[U_i,V_i] \Big ),
$$
where $t$ is arbitrary, $|X_i| \leq \lfloor \sqrt{k} \rfloor$, $k(n) \geq 5$, and all rectangles $U_i \times V_i \subseteq \mathcal{B}$, for some set $\mathcal{B} \subseteq U_{n,k} \times V_{n,k}$ of locality $\mu \leq 1/16$. Then, for large enough $n$, the following holds.
\begin{itemize} 
\item[\emph{1.}] Either $C(x, \mathcal{F}^{\star})$ accepts a subset of $V_{n,k}$ of measure at least $1/10$, or
\item[\emph{2.}] $C(x, \mathcal{F}^{\star})$ rejects a subset of $U_{n,k}$ of measure at least $1/10$.
\end{itemize}  
\end{lemma}

\begin{proof}
If $t = 0$ the circuit computes a constant function, and consequently one of the items above must hold. Otherwise, for each $i \in [t]$, since $U_i \times V_i \subseteq \mathcal{B}$ and $\mathcal{D}_{n,k} = \mathcal{D}^U_{n,k} \times \mathcal{D}^V_{n,k}$, we have that either $\mathcal{D}_{n,k}^U(U_i) \leq \mu^{1/2}$ or $\mathcal{D}^V_{n,k}(V_i) \leq \mu^{1/2}$. We consider two cases.

First, assume there is $i \in [t]$ such that $\mathcal{D}^{V}_{n,k}(V_i) \leq \mu^{1/2} \leq 1/4$. Then,
\begin{equation}
\Pr_{H \sim \mathcal{D}^V_{n,k}}[(\lceil X_i \rceil \wedge f^\star_i)(H) = 1] \,\geq\, 1 - \Pr[\lceil X_i \rceil(H) = 0]  - \Pr[H \in V_i] \,\geq\, 3/4 - \Pr[\lceil X_i \rceil(H) = 0]. \nonumber 
\end{equation}
The latter probability is $0$ if $|X_i| \leq 1$. Otherwise, it can be upper bounded by
\begin{eqnarray}
\Pr_{\chi \sim V^\chi_{n,k}}[\,|\chi(X_i)| < |X_i|\,\mid\,G(\chi) \in V_{n,k}\,] & \leq & (1 + o(1)) \cdot \sum_{\{a,b\} \in \binom{X_i}{2}} \Pr_{\chi \sim V^\chi_{n,k}}[\chi(a) = \chi(b)] \nonumber \\
(\text{since}~|X_i| \leq \lfloor \sqrt{k} \rfloor) & \leq & (1 + o(1)) \cdot \binom{\lfloor \sqrt{k} \rfloor}{2} \cdot \frac{k-1}{(k-1)^2}. \nonumber
\end{eqnarray}
This shows that item $1$ above holds, using $k \geq 5$ and the previous estimate. 

If there is no $i \in [t]$ satisfying $\mathcal{D}^{V}_{n,k}(V_i) \leq \mu^{1/2}$, by the observation in the first paragraph of this proof we get that  $\mathcal{D}^U_{n,k}(U_i) \leq \mu^{1/2}$ and $\mathcal{D}^V_{n,k}(V_i) > \mu^{1/2}$ for all $i \in [t]$. Recall that the measure of $\mathcal{B}$ is at most $\mu \leq 1/16$. Therefore, it must be the case that $|\bigcup_i U_i|/|U_{n,k}| \leq \mu^{1/2}$, as each $K_B$ in this union contributes at least $\mu^{1/2}$ to the measure of $\mathcal{B}$. Due to our choice of $\mathcal{F}^{\star}$ and the structure of $C$, $C(\vec{x}, \mathcal{F}^\star)$ will accept at most a $(1/4)$-fraction of $U_{n,k}$, and item 2 holds.
\end{proof}

Crucially, Lemma \ref{l:lb-depth2} requires no upper bound on the number of terms appearing in $C$, and this will play a fundamental role in the argument below.

For the rest of the proof, let $D(\vec{x}, \vec{y})$ be a monotone CLO of size $s$ that computes $k$-clique over $A_{n,k}$, and $W_i = (V_i, U_i)$ for $i \leq e$ be its associated pairs, where $e \leq s$. As usual, we set $\mathcal{B} = \bigcup_i U_i \times V_i$. Recall the extra condition on the local oracle gates.\\

\noindent \textbf{Assumption $\mathcal{A}_d$:} If $J \subseteq [e]$ and $|J| > d$, then $\bigcap_{j \in J} U_j = \emptyset$.\\

\blue{We can assume without loss of generality that different oracle variables appearing in the description of the circuit are associated to distinct subsets of $U_{n,k}$. Indeed, due to monotonicity (cf.~Claim \ref{cl:hardF}), we can always take a larger subset of $V_{n,k}$ if different oracle variables are associated to the same subset of $U_{n,k}$. A bit more precisely, if $y_i = y_i[U', V_i]$ and $y_j = y_j[U', V_j]$, we can redefine these local oracles to use the pair $(U', V_i \cup V_j)$. This does not increase the overall locality, and does not change the correctness of the computation. Note that this transformation produces oracle variables associated to the same pair of subsets, but since we use boolean circuits instead of boolean formulas, oracle variables don't need to be repeated in the description of the circuit.}

For $J \subseteq [e]$, we use $D_J(\vec{x})$ to denote the circuit with $y_j$ substituted by $1$ if $j \in J$, and by $0$ otherwise. In particular, each $D_J$ is a monotone circuit in the usual sense, i.e., it does not contain local oracle gates. Moreover, $\mathsf{size}(D_J) \leq \mathsf{size}(D)$.

\begin{lemma}\label{l:simplification}
Under Assumption $\mathcal{A}_d$, for every input graph $G \in A_{n,k}$, 
$$
D(G,\mathcal{F}^\star) = \bigvee_{J \in \binom{[e]}{\leq d}} D_J(G) \wedge f^\star_{(U_J,V_J)}(G),
$$
where $U_J \eqdef \bigcap_{j \in J} U_j$ and $V_J \eqdef \bigcup_{j \in J} V_j$ \emph{(}here an empty intersection is $U_{n,k}$ and an empty union is $\emptyset$, corresponding to the case where $J = \emptyset$\emph{)}. 
\end{lemma}

\begin{proof}
First, observe that for inputs in $A_{n,k}$,
$$
D(\vec{x}, \mathcal{F}^\star) \equiv \bigvee_{J \subseteq [e]} \Big ( D_J(\vec{x}) \wedge \bigwedge_{j \in J} f^\star_j(\vec{x}) \wedge \bigwedge_{j \notin J} \neg f^\star_j(\vec{x}) \Big )\;,
$$
using our definition of $D_J(\vec{x})$. As we explain below, this circuit is further equivalent to a circuit where we drop the negated part:
$$
D(\vec{x}, \mathcal{F}^\star) \equiv \bigvee_{J \subseteq [e]} \Big ( D_J(\vec{x}) \wedge \bigwedge_{j \in J} f^\star_j(\vec{x}) \Big )\,.
$$
Clearly, by eliminating some ``literals'' we can only accept more inputs. However, by monotonicity the latter is not going to happen. Indeed, if we have a term and a negative input $H \in V_{n,k}$ such that $D_J(H) \wedge \bigwedge_{j \in J} f^\star_j(H) = 1$ but $\bigwedge_{j \notin J} \neg f^\star_j(H) = 0$, then there is a set $J'$ with $J \subseteq J' \subseteq [e]$ such that $D_{J'}(H) \wedge \bigwedge_{j \in J'} f^\star_j(H) \wedge \bigwedge_{j \notin J'} \neg f^\star_j(H) = 1$, where we have used the monotonicity of $D(\vec{x}, \vec{y})$ in order to claim that $D_{J'}(H) \geq D_J(H)$. This is impossible, since by assumption $D(\vec{x}, \mathcal{F}^\star)$ separates $U_{n,k}$ and $V_{n,k}$.

Using Lemma \ref{l:lemmaK}, we know that $\bigwedge_{j \in J} f^\star_j = f^\star_{(U_J, V_J)}$, for $U_J$ and $V_J$ as in the statement of the lemma. Under assumption $\mathcal{A}_d$, whenever $|J| > d$ we get $U_J = \emptyset$. Therefore,
\begin{equation} \label{eq:normalform}
D(\vec{x}, \mathcal{F}^\star) \equiv \bigvee_{J \in  \binom{[e]}{\leq d}} \Big ( D_J(\vec{x}) \wedge f^\star_{(U_J,V_J)}(\vec{x}) \Big ) \vee \bigvee_{J \in \binom{[e]}{> d}} \Big ( D_{J}(\vec{x}) \wedge f^\star_{(\emptyset, V_{J})}(\vec{x}) \Big)\;.
\end{equation}

Using the equivalences established above and the correctness of the original circuit, the circuit in (\ref{eq:normalform}) accepts every input in $U_{n,k}$, and rejects every input in $V_{n,k}$. Now observe that the right-hand terms of the circuit cannot accept an input in $V_{n,k}$, due to the presence of the functions $f^\star_{(\emptyset, V_J)}$. Thus such terms can be discarded, and the circuit obtained after this simplification still accepts $U_{n,k}$ and rejects $V_{n,k}$. This completes the proof of the lemma.
\end{proof}

Observe that $U_J \times V_J \subseteq \mathcal{B}$ for every $J \subseteq [e]$, due to Lemma \ref{l:lemmaK}. In particular, the simplification above is well-behaved with respect to the new oracle rectangles introduced in the transformation.

The next steps of our argument rely on results from Alon and Boppana \citep{DBLP:journals/combinatorica/AlonB87} related to the approximation method \citep{razborov1985lower}. We follow the terminology of the exposition in Boppana and Sipser \citep[Section 4.2]{DBLP:books/el/leeuwen90/BoppanaS90}. For the rest of the proof, we let $\ell \eqdef \lfloor \sqrt{k} \rfloor$, $p \eqdef \lceil 10 \sqrt{k} \log n \rceil$, and $m \eqdef (p - 1)^\ell \cdot \ell!$. (Recall that $\ell$ is the size of each indicator set $\lceil X_i \rceil$, $m$ is the maximum number of indicators in each approximator, and $p$ is an auxiliary parameter.\footnote{Do not confuse this definition of $m$ with the number of edges in the input graph, which will not be needed in the rest of the proof.})

Approximate each individual circuit $D_J(\vec{x})$ as in Boppana-Sipser, obtaining a corresponding depth-$2$ approximator $\widetilde{D}_J(\vec{x})$. Since each $D_J(\vec{x})$ is a monotone circuit of size at most $s$, our choice of $U_{n,k}$ and $V_{n,k}$ and the argument in \citep{DBLP:books/el/leeuwen90/BoppanaS90} provide the following bounds.

\begin{lemma}\emph{\citep[Lemma 4.3]{DBLP:books/el/leeuwen90/BoppanaS90}.}\label{l:l1}
For each $J \subseteq [e]$, the number of positive test graphs $G \in U_{n,k}$ for which $D_J(G) \leq \widetilde{D}_J(G)$ does not hold is at most $E^+ \eqdef s \cdot m^2 \cdot \binom{n - \ell - 1}{k - \ell - 1}$.
\end{lemma}

\vspace{0.1cm}

\begin{lemma}\emph{\citep[Lemma 4.4]{DBLP:books/el/leeuwen90/BoppanaS90}.}\label{l:l2}
For each $J \subseteq [e]$, the number of negative test graphs \emph{(}colorings\emph{)} $\chi \in V^\chi_{n,k}$ for which $D_J(G(\chi)) \geq \widetilde{D}_J(G(\chi))$ does not hold is at most $E^- \eqdef s \cdot m^2 \cdot [\binom{l}{2}/(k-1)]^p \cdot (k-1)^n$.
\end{lemma}

Now define using $D$ and the individual approximators $\widetilde{D}_J$ a corresponding monotone circuit $\widetilde{D}(\vec{x}, \vec{y})$ with access to the functions $f^\star_{(U_J, V_J)}$: 
\begin{equation}\label{eq:approx}
\widetilde{D}(\vec{x},\mathcal{F}^\star) \eqdef \bigvee_{J \in \binom{[e]}{\leq d}} \big ( \widetilde{D}_J(\vec{x}) \wedge f^\star_{(U_J,V_J)}(\vec{x}) \big ).
\end{equation}
Clearly, $D(G, \mathcal{F}^\star) \neq \widetilde{D}(G, \mathcal{F}^\star)$ on an input $G \in A_{n,k}$ only if for some approximator $\widetilde{D}_J$ we have $\widetilde{D}_J(G) \neq D_J(G)$. Furthermore, at most $\sum_{j = 0}^d \binom{e}{j} \leq (e+1)^d \leq (s + 1)^d$ distinct circuits $D_J$ are approximated. Combining this with Lemmas \ref{l:l1} and \ref{l:l2}, a union bound, and the fact that the original circuit is correct on every input graph in $A_{n,k}$, we get:
\begin{equation}
\Pr_{G \sim \mathcal{D}^U_{n,k}}[\widetilde{D}(G,\mathcal{F}^\star) = 1]  \;\geq\;  1 - (s + 1)^d \cdot \frac{E^+}{\binom{n}{k}}, \nonumber
\end{equation}
and similarly,
\begin{eqnarray}
\Pr_{H \sim \mathcal{D}^V_{n,k}}[\widetilde{D}(H,\mathcal{F}^\star) = 0] & \geq & (1 - o(1)) \cdot \Pr_{\chi \sim V^\chi_{n,k}}[\widetilde{D}(G(\chi),\mathcal{F}^\star) = 0 \wedge G(\chi) \in V_{n,k}] \nonumber \\
& \geq & (1 - o(1))\cdot \big ( 1 - \Pr_{\chi}[\widetilde{D}(G(\chi),\mathcal{F}^\star) = 1] - o(1) \big ) \nonumber \\
& \geq & (1 - o(1)) \cdot \Big (1 - (s + 1)^d \cdot \frac{E^-}{(k-1)^n} \Big ). \nonumber 
\end{eqnarray}

We can assume each one of these probabilities $\to 1$ as $n \to \infty$, since otherwise we get that $s \geq n^{\Omega(\sqrt{k}/d)}$ using the values of $E^-$, $E^+$, $p$, $\ell$, and $m$, completing the proof of Theorem \ref{t:thm2}. In more detail, let $\delta > 0$ be an arbitrary small constant, and suppose that:
$$
(s+1)^d \cdot \frac{s \cdot m^2 \cdot \binom{n - \ell - 1}{k - \ell - 1}}{\binom{n}{k}} \;\geq\; \delta  \quad \text{or} \quad  (s + 1)^d \cdot \frac{s \cdot m^2 \cdot [\binom{l}{2}/(k-1)]^p \cdot (k-1)^n}{(k-1)^n} \;\geq \; \delta.
$$ 
Due to the upper bound on $k$ in the statement of Theorem \ref{t:thm2}, using estimates entirely analogous to the ones employed in \citep{DBLP:books/el/leeuwen90/BoppanaS90} (which are routine and left to the reader), it follows in each case that:
$$
(s + 1)^{d+1} \;\geq\; n^{\Omega(\sqrt{k})}.
$$
This justifies the claim made above on the convergence of the probabilities.

Now expand each term $\widetilde{D}_J(\vec{x}) \wedge f^\star_{(U_J, V_J)}(\vec{x})$ in $\widetilde{D}(\vec{x}, \mathcal{F}^\star)$ (Equation \ref{eq:approx}), using that (see \citep{DBLP:books/el/leeuwen90/BoppanaS90}) \blue{each circuit $\widetilde{D}_J(\vec{x})$ is either a union of clique indicators of bounded size:}
$$
\widetilde{D}_J(\vec{x}) \equiv \bigvee_{i \in [m_J]} \lceil X^J_i \rceil \;
$$ 
for $m_J \leq m$ and an appropriate choice of sets $X^J_i \subseteq [n]$ satisfying $0 \leq |X^J_i| \leq \ell$, \blue{or $\widetilde{D}_J \equiv 0$.} This produces a circuit equivalent to $\widetilde{D}(\vec{x}, \mathcal{F}^\star)$ over inputs in $A_{n,k}$, and it can be written in the following form:
\begin{equation}\label{eq:finalcircuit}
\widetilde{D}(\vec{x},\mathcal{F}^\star) \equiv \bigvee_{i \in [t]} \Big ( \lceil X_i \rceil \wedge f^\star_{(U'_i,V'_i)}(\vec{x}) \Big )
\end{equation}
Here $t$ can be arbitrarily large, but observe that $U'_i \times V'_i \subseteq \mathcal{B}$ for every $i \in [t]$ (due to Lemmas \ref{l:lemmaK} and \ref{l:simplification}). We don't assume that $(U'_i,V'_i) \neq (U'_{i'}, V'_{i'})$ when $i \neq i'$, and similarly for $X_i$ and $X_{i'}$.

Finally, we know that the circuit in Equation \ref{eq:finalcircuit} accepts a subset of $U_{n,k}$ of measure $1 - o(1)$, and that it rejects a subset of $V_{n,k}$ of measure $1 - o(1)$. By construction, each clique indicator in the description of $\widetilde{D}$ has size at most $\ell \leq \lfloor \sqrt{k} \rfloor$. Together with $U'_i \times V'_i \subseteq \mathcal{B}$ for every $i \in [t]$ and the upper bound on the locality of $\mathcal{B}$, we get a contradiction to  Lemma \ref{l:lb-depth2}. 

The proof of Theorem \ref{t:thm2} is complete. Observe that, under the same assumptions, it is possible to obtain a slightly stronger trade-off of the form: $e^d \cdot s \;\geq\; n^{\Omega(\sqrt{k})}$.

\section{Concluding remarks}\label{s:concluding}

We discuss below some questions and directions motivated by our results, and elaborate a bit more on the connection to proof complexity.\\

\noindent \textbf{Monotone circuit complexity.} The main open problem in the context of circuit complexity is to understand the  size of monotone circuits of small locality separating the sets $U_{n,k}$ and $V_{n,k}$, under no further assumption on the $y$-variables. It is not clear if the hypothesis $\mathcal{A}_d$ in Theorem \ref{t:thm2} is  an artifact of our proof. As far as we know, it is conceivable that smaller circuits can be designed by increasing the overlap between the sets $U_i$.\footnote{We notice that non-monotone polynomial size circuits containing oracles of small locality can compute any boolean function (see \citep[Section 3]{1611.08680}). A similar phenomenon appears in the adaptation of real-valued monotone circuits to general real-valued circuits \citep[Section 7]{pudlak1997lower}, but in that case strong lower bounds are known against monotone real-valued circuits.} 

However, if one is more inclined to lower bounds, we mention that the fusion approach described in \cite{DBLP:conf/coco/Karchmer93} can be easily adapted to monotone circuit with local oracles, and that this point of view might be helpful in future investigations of unrestricted monotone CLOs. 

Another question of combinatorial interest is whether the phase transitions observed in Theorem \ref{t:thm1} extend to more expressive classes of monotone circuits beyond depth two. More broadly, are the phase transitions observed here particular to $k$-clique, or an instance of a more general phenomenon connected to computations using monotone circuits extended with oracle gates?

Corollary \ref{c:tight} suggests the following problem. Is it possible to refine the approach from  \citep{DBLP:journals/combinatorica/AlonB87}, and to prove that the monotone circuit size complexity of $k$-clique is $n^{\Omega(k)}$ for a larger range of $k$? In a related direction, it would be interesting to understand if monotone CLOs can shed light into the difficulties in proving stronger monotone circuit size lower bounds for other boolean functions of interest, such as the matching problem on graphs (see e.g.~\citep[Section 5]{DBLP:journals/combinatorica/AlonB87} and \citep[Section 9.11]{DBLP:books/daglib/0028687}).\\

\noindent \textbf{Proof complexity.} Back to the original motivation from proof complexity, we have been unable so far to transform proofs in R$(\mathsf{Lin}/\mathbb{F}_2)$ into monotone CLOs satisfying $\mathcal{A}_d$, for $d \leq k^{1/2 - \varepsilon}$, or certain variations of $\mathcal{A}_d$ under which Theorem \ref{t:thm2} still holds. Observe that, using the connections established in \citep{1611.08680}, this would be sufficient for exponential lower bounds on proof size.   

The reduction from randomized feasible interpolation
actually provides a distribution on monotone CLOs $C_r$
with a common bound on their sizes such that each
is correct and they satisfy:
$$
\Pr_r [ (u,v) \in \mathcal{B}_r ] \le \mu ~\; \text{for every fixed pair~} (u,v) \in U \times V,
$$
where $\mathcal{B}_r$ is the union of the oracle rectangles in $C_r$.
An averaging argument then yields a fixed monotone CLO whose locality is bounded by $\mu$. One might lose
some information useful for a lower bound in this last step depending on the choice of the distribution $\mathcal{D}$ supported over $U \times V$.

Even though our initial attempts at establishing new length-of-proofs lower bounds have been unsuccessful, we feel that in order to prove limitations for R$(\mathsf{Lin}/\mathbb{F}_2)$ and for other proof systems via randomized feasible interpolation it should be sufficient to establish lower bounds against monotone CLOs under an appropriate assumption on the oracle gates. (In particular, the existence of monotone CLOs of small size and small locality separating $U_{n,k}$ and $V_{n,k}$ does not imply that the approach presented in \citep{1611.08680} is fruitless.) For instance, while $\mathcal{A}_d$ is a semantic condition on the (unstructured) sets $U_i$ and $V_i$, one can try to  explore the syntactic information obtained on these sets from a given proof, such as upper bounds on the circuit complexity of separating each pair $U_i$ and $V_i$, or other related structural information.\\

\noindent \textbf{Acknowledgements.} We would like to thank Pavel Pudl\'{a}k for discussions on monotone circuits with local oracles and proof complexity. The second author would like to thank Michal Garl\'{i}k for several related conversations. This work was supported by the European Research Council under the European Union’s Seventh Framework Programme (FP7/2007-2013)/ERC Grant No.~615075.

\bibliographystyle{alpha}	
\bibliography{refs}

\begin{thebibliography}{BPR97}

\bibitem[AB87]{DBLP:journals/combinatorica/AlonB87}
Noga Alon and Ravi~B. Boppana.
\newblock The monotone circuit complexity of boolean functions.
\newblock {\em Combinatorica}, 7(1):1--22, 1987.

\bibitem[BKZ15]{buss2015collapsing}
Samuel Buss, Leszek Ko{\l}odziejczyk, and Konrad Zdanowski.
\newblock Collapsing modular counting in bounded arithmetic and constant depth
  propositional proofs.
\newblock {\em Transactions of the American Mathematical Society},
  367(11):7517--7563, 2015.

\bibitem[BPR97]{DBLP:journals/jsyml/BonetPR97}
Maria~Luisa Bonet, Toniann Pitassi, and Ran Raz.
\newblock Lower bounds for cutting planes proofs with small coefficients.
\newblock {\em J. Symbolic Logic}, 62(3):708--728, 1997.

\bibitem[BS90]{DBLP:books/el/leeuwen90/BoppanaS90}
Ravi~B. Boppana and Michael Sipser.
\newblock The complexity of finite functions.
\newblock In {\em Handbook of Theoretical Computer Science, Volume {A:}
  Algorithms and Complexity}, pages 757--804. 1990.

\bibitem[IS14]{ItsykonSokolov}
Dmitry Itsykson and Dmitry Sokolov.
\newblock Lower bounds for splittings by linear combinations.
\newblock In {\em Mathematical Foundations of Computer Science \emph{(MFCS)}},
  pages 372--383, 2014.

\bibitem[Juk12]{DBLP:books/daglib/0028687}
Stasys Jukna.
\newblock {\em Boolean Function Complexity - Advances and Frontiers}.
\newblock Springer, 2012.

\bibitem[Kar93]{DBLP:conf/coco/Karchmer93}
Mauricio Karchmer.
\newblock On proving lower bounds for circuit size.
\newblock In {\em Structure in Complexity Theory Conference \emph{(CCC)}},
  pages 112--118, 1993.

\bibitem[Kra95]{krajicek}
Jan Kraj\'{\i}\v{c}ek.
\newblock {\em Bounded Arithmetic, Propositional Logic, and Complexity Theory}.
\newblock Cambridge University Press, 1995.

\bibitem[Kra97]{Krajicek97}
Jan Kraj\'{\i}\v{c}ek.
\newblock Interpolation theorems, lower bounds for proof systems, and
  independence results for bounded arithmetic.
\newblock {\em J. Symbolic Logic}, 62(2):457--486, 1997.

\bibitem[Kra16]{1611.08680}
Jan Kraj\'{\i}\v{c}ek.
\newblock Randomized feasible interpolation and monotone circuits with a local
  oracle.
\newblock {\em \emph{Available at arXiv:1611.08680}}, 2016.

\bibitem[Pud97]{pudlak1997lower}
Pavel Pudl{\'a}k.
\newblock Lower bounds for resolution and cutting plane proofs and monotone
  computations.
\newblock {\em J. Symbolic Logic}, 62(3):981--998, 1997.

\bibitem[Raz85]{razborov1985lower}
Alexander~A. Razborov.
\newblock Lower bounds on the monotone complexity of some boolean functions.
\newblock {\em Soviet Math. Doklady}, 31:354--357, 1985.

\bibitem[Rob13]{RRobere}
Robert Robere.
\newblock Average-case lower bounds for monotone switching networks, 2013.
\newblock (Masters thesis, University of Toronto).

\bibitem[Ros14]{DBLP:journals/siamcomp/Rossman14}
Benjamin Rossman.
\newblock The monotone complexity of $k$-clique on random graphs.
\newblock {\em {SIAM} J. Comput.}, 43(1):256--279, 2014.

\end{thebibliography}

\vspace{0.7cm}
\noindent \textbf{Mailing address:}

Department of Algebra

Faculty of Mathematics and Physics

Charles University

Sokolovsk\'{a} 83, Prague 8, CZ -- 186 75

The Czech Republic

\end{document}